\documentclass[journal]{IEEEtran}
\usepackage{amsfonts}
\usepackage{amssymb}
\usepackage{amsthm}
\usepackage{amsmath,amsfonts,amssymb}
\usepackage[dvips]{graphicx}
\usepackage{verbatim}
\usepackage{setspace}
\usepackage{bm}
\usepackage{algorithmic} 
\usepackage[ruled,vlined]{algorithm2e}
\usepackage{cite}

\usepackage{changepage}
\usepackage{pdfpages}
\usepackage{color}
\newtheorem{theorem}{Theorem}
\newtheorem{lemma}{Lemma}

\newcommand{\figwidth}{8}
\IEEEoverridecommandlockouts

\begin{document}
%
\title{User Fairness Non-orthogonal Multiple Access (NOMA) for 5G Millimeter-Wave Communications with Analog Beamforming}
%
%
%
\author{
        Zhenyu Xiao,~\IEEEmembership{Senior Member,~IEEE,}
        Lipeng Zhu,
        Zhen Gao,~\IEEEmembership{Member,~IEEE,}
        Dapeng Oliver Wu,~\IEEEmembership{Fellow,~IEEE}
        and Xiang-Gen Xia,~\IEEEmembership{Fellow,~IEEE}
\thanks{Z. Xiao and L. Zhu are with the School of
Electronic and Information Engineering, Beihang University, Beijing 100191, China}
\thanks{Z. Gao is with the Advanced Research Institute of Multidisciplinary Science, Beijing Institute of Technology, Beijing 100081, China.}
\thanks{D. O. Wu is with the Department of Electrical and Computer Engineering, University of Florida, Gainesville, FL 32611, USA.}
\thanks{X.-G. Xia is with the Department of Electrical and Computer Engineering, University of Delaware, Newark, DE 19716, USA.}
}

%
%

\maketitle

\begin{abstract}
The integration of non-orthogonal multiple access in millimeter-Wave communications (mmWave-NOMA) can significantly improve the spectrum efficiency and increase the number of users in the fifth-generation (5G) mobile communication. In this paper we consider a downlink mmWave-NOMA cellular system, where the base station is mounted with an analog beamforming phased array, and multiple users are served in the same time-frequency resource block. To guarantee user fairness, we formulate a joint beamforming and power allocation problem to maximize the minimal achievable rate among the users, i.e., we adopt the max-min fairness. As the problem is difficult to solve due to the non-convex formulation and high dimension of the optimization variables, we propose a sub-optimal solution, which makes use of the spatial sparsity in the angle domain of the mmWave channel. In the solution, the closed-form optimal power allocation is obtained first, which reduces the joint optimization problem into an equivalent beamforming problem. Then an appropriate beamforming vector is designed. Simulation results show that the proposed solution can achieve a near-upper-bound performance in terms of achievable rate, which is significantly better than that of the conventional mmWave orthogonal multiple access (mmWave-OMA) system.
\end{abstract}


%
\IEEEpeerreviewmaketitle

\section{Introduction}
\IEEEPARstart{W}{ith} the rapid growth of mobile data traffic, higher data rate is an insistent requirement in the fifth generation (5G) of mobile communication \cite{andrews2014will}. Millimeter-Wave (mmWave) communications, with frequency ranging from 30-300 GHz, provides abundant spectrum resources and is perceived as a candidate key technology for 5G \cite{andrews2014will,rapp2013mmIEEEAccess,XiaoM2017survmmWave}. In addition to the large amount of bandwidth, the mmWave-band signal has a shorter wavelength compared with the traditional microwave-band signal, which makes it possible to equip a large antenna array in a small area. Considerable beam gain can be obtained to overcome the high propagation loss in the mmWave-band \cite{XiaoM2017survmmWave}.

Although more spectrum resources are available in the mmWave band, multiple access is still an important issue to increase the spectrum efficiency and the number of users/devices to support 5G Internet of Things (IoT). Non-orthogonal multiple access (NOMA), considered as another candidate technology for 5G, has drawn widespread attention in both academia and industry \cite{Ding2017mmWaveNOMA,Benjebbour2013ConceptNOMA,saito2013non,ding2014performance,Dai2015NOMA5G,Ding2015Cooperative,Ding2017survNOMA,Zhu2018optimaluserp}. Different from the conventional orthogonal multiple access (OMA) schemes, NOMA serves multiple users in a single resource block (time/frenquency/code) and distinguishes them in power domain. Successive interference cancellation (SIC) is required at the receivers. In general, the users are sorted by an increasing order of channel gains. The one with lower channel gain is prior, i.e., its signal is decoded and removed first with the signals of the other users treated as noise \cite{Benjebbour2013ConceptNOMA,saito2013non,ding2014performance,Dai2015NOMA5G,Ding2015Cooperative,Ding2017survNOMA}. In this way, NOMA can increase the spectrum efficiency and break the limit that the maximal number of users is no larger than the number of radio-frequency (RF) chains in OMA networks \cite{ding2014performance,Dai2015NOMA5G,Ding2015Cooperative,Ding2017random,Daill2017,roh2014millimeter,sun2014mimo}.

To make full use of the spectrum resource, we investigate NOMA in mmWave communications (mmWave-NOMA) in this paper. The combination of the two candidate technologies for 5G has been preliminarily explored in several literatures. In \cite{Ding2017random}, the coexistence of NOMA and mmWave communications was considered, where random beamforming was used in order to reduce the system overhead. The results demonstrated that the combination of NOMA and mmWave communications yields significant gains in terms of sum rates and outage probabilities, compared with the conventional mmWave-OMA systems. In \cite{Daill2017}, the new concept of beamspace multiple-input multiple-output NOMA (MIMO-NOMA) with a lens-array hybrid beamforming structure was proposed to use multi-beam to serve multiple NOMA users with arbitrary locations. With this method, the number of supported users can be larger than the number of RF chains in the same time-frequency resource block. Beamforming, user selection and power allocation were considered for mmWave-NOMA networks in \cite{Cui2018mmWaveNOMA}, where random beamforming was adopted first. Then a power allocation algorithm that leverages the branch and bound (BB) technique and a low complexity user selection algorithm based on matching theory were proposed. A NOMA based hybrid beamforming design was proposed in \cite{Wu2017hybridBF}, where a user pairing algorithm  was proposed first and then the hybrid beamforming and power allocation algorithm was proposed to maximize the sum achievable rate. In \cite{Zhang2017mmWaveMIMONOMA}, the NOMA-mmWave-massive-MIMO system model and a simplified mmWave channel model were proposed. Whereafter, theoretical analysis on the achievable rate was considered in both the noise-dominated low-SNR regime and the interference-dominated high-SNR regime. To further improve the data rate, power allocation and beamforming were jointly explored in \cite{xiao2018mmWaveNOMA} and \cite{Zhulip2018uplink} for a 2-user downlink and uplink mmWave-NOMA scenario, respectively, where the key technique is the multi-directional beamforming design with a constant-modulus (CM) phased array.

Different from these works \cite{Ding2017random,Cui2018mmWaveNOMA,Daill2017,Zhang2017mmWaveMIMONOMA,Wu2017hybridBF,xiao2018mmWaveNOMA}, we consider user fairness for downlink mmWave-NOMA networks in this paper. To improve the overall data rate, we maximize the minimal achievable rate among multiple users, i.e., we adopt the max-min fairness \footnote{We adopt the max-min fairness because it is a typical and extensively used fairness rule in NOMA \cite{FairnessNOMA2015,Ding2017survNOMA}. Besides the max-min fairness, there are also other fairness rules in NOMA, like proportional fairness, etc. \cite{Ding2017survNOMA}.}. Due to the requirement of low hardware cost and power consumption, an analog beamforming structure with a single RF chain is utilized, where both implementations of single phase shifter (SPS) and double phase shifter (DPS) are considered \cite{Bogale2016DPS,Lin2017DPS}. In the formulated problem, power allocation and beamforming are jointly optimized. As the problem is non-convex and the dimension of the optimization variables is large, it is difficult to solve this problem with the existing optimization tools. To this end, we solve this problem with two stages and obtain a sub-optimal solution. In the first stage, we obtain closed-form optimal power allocation with an arbitrary fixed beamforming vector, which reduces the joint optimization problem into an equivalent beamforming problem. Then, in the second stage, we propose an appropriate beamforming algorithm utilizing the spatial sparsity in the angle domain of the mmWave channel. Finally, we verify the performance of the proposed joint beamforming and power allocation method for user fairness mmWave-NOMA by simulations. The results show that the proposed solution can achieve a near-upper-bound performance in terms of achievable rate, which is significantly better than that of the conventional mmWave-OMA system.

The rest of the paper is organized as follows. In Section II, we present the system model and formulate the problem. In Section III, we propose the solution. In Section IV, simulation results are given to demonstrate the performance of the proposed solution, and the paper is concluded finally in Section V.

Symbol Notation: $a$ and $\mathbf{a}$ denote a scalar variable and a vector, respectively. $(\cdot)^{\rm{T}}$ and $(\cdot)^{\rm{H}}$ denote transpose and conjugate transpose, respectively. $|\cdot|$ and $\|\cdot\|$ denote the absolute value and Euclidean norm, respectively. $\mathbb{E}(\cdot)$ denotes the expectation operation. $[\mathbf{a}]_i$ denotes the $i$-th entry of $\mathbf{a}$. $\mathbb{C}^{N}$ denotes an $N$-dimension linear space in complex domain.

\section{System Model and Problem Formulation}
\subsection{System model}
\begin{figure}[t]
\begin{center}
  \includegraphics[width=\figwidth cm]{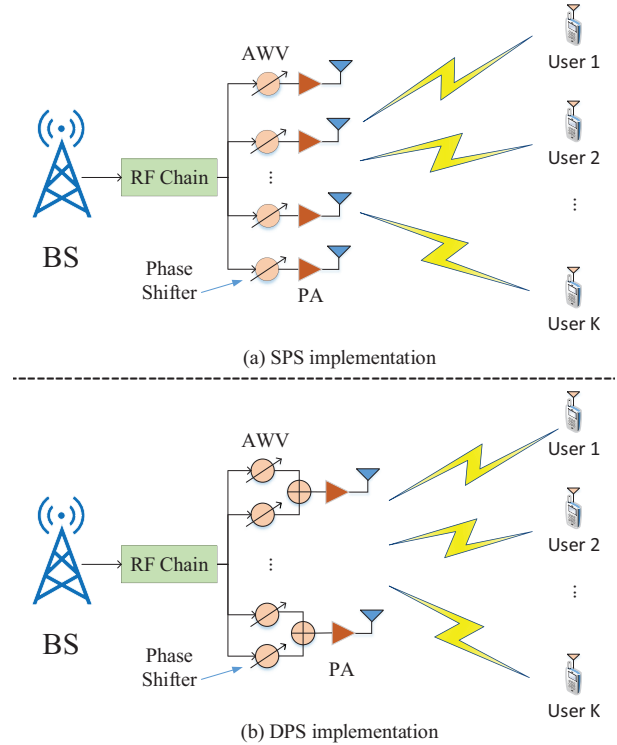}
  \caption{Illustration of a mmWave mobile cell, where one BS with $N$ antennas serves multiple users with one single antenna.}
  \label{fig:system}
\end{center}
\end{figure}
In this paper, we consider a downlink mmWave communications system. As shown in Fig. \ref{fig:system}, the base station (BS) is equipped with a single RF chain and an $N$-antenna phased array. $K$ users with a single antenna are served in the same resource block. Each antenna is driven by the power amplifier (PA) and phase shifter (PS).

The BS transmits a signal $s_k$ to User $k~(k=1,2,\cdots, K)$ with transmission power $p_k$, where $\mathbb{E}(\left | s_{k} \right |^{2})=1$. The total transmission power of the BS is $P$. The received signal for User $k$ is
\begin{equation}
y_k=\mathbf{h}_{k}^{\rm{H}}\mathbf{w}\sum \limits_{k=1}^{K}\sqrt{p_{k}}s_{k}+n_k,
\end{equation}
where $\mathbf{h}_{k}$ is the channel response vector between the BS and User $k$. $\mathbf{w}$ is the antenna weight vector (AWV), i.e., analog beamforming vector, and $n_k$ denotes the Gaussian white noise at User $k$ with power $\sigma^2$.

Two PS structures, named SPS implementation and DPS implementation, are considered. For the SPS implementation, each antenna branch has a single PS as shown in Fig. \ref{fig:system}(a). The elements of the AWV are complex numbers, whose modulus and phase are controlled by the PA and PS respectively. To reduce hardware complexity, all the PAs have the same scaling factor in general. Thus, the AWV has CM elements, which is denoted by
\begin{equation}\label{CM}
|[{\mathbf{w}}]_i|=\frac{1}{\sqrt{N}},~i=1,2,...,N.
\end{equation}

The above constraint is non-convex, which results in a major challenge of AWV design, i.e., we can only adjust the phase but not the amplitude of the signal. To reduce the design difficulty, a new implementation named DPS was proposed in \cite{Bogale2016DPS,Lin2017DPS}, which is shown in Fig. \ref{fig:system}(b). For the DPS implementation, each antenna is driven by the summation of the two independent PSs. Although the modulus of each PS is constant, the phases of two PSs can be adjusted to achieve different modulus in each antenna branch. Thus, the modulus constraint is relaxed to
\begin{equation}\label{DPS}
|[{\mathbf{w}}]_i| \leq \frac{2}{\sqrt{N}},~i=1,2,...,N.
\end{equation}
By doubling the number of PSs, the new constraint becomes convex and therefore make it more tractable to develop low-complexity design approaches.

The channel between the BS and User $k$ is a mmWave channel.\footnote{In this paper, we assume the channel state information (CSI)  is known by the BS. The mmWave channel estimation with low complexity can be referred in \cite{xiao2016codebook} and \cite{xiao2017codebook}.} Subject to the limited scattering in mmWave-band, multipath is mainly caused by reflection. As the number of the multipath components (MPCs) is small in general, the mmWave channel has directionality and appears spatial sparsity in the angle domain \cite{peng2015enhanced,wang2015multi,Lee2014exploiting,Gao2016ChannelEst,xiao2016codebook,alkhateeb2014channel}. Different MPCs have different angles of departure (AoDs). Without loss of generality, we adopt the directional mmWave channel model assuming a uniform linear array (ULA) with a half-wavelength antenna space. Then a mmWave channel can be expressed as \cite{peng2015enhanced,wang2015multi,Lee2014exploiting,Gao2016ChannelEst,xiao2016codebook,alkhateeb2014channel}
\begin{equation} \label{eq_oriChannel}
\mathbf{h}_{k}=\sum_{\ell=1}^{L_k}\lambda_{k,\ell}\mathbf{a}(N,\Omega_{k,\ell}).
\end{equation}
where $\lambda_{k,\ell}$, $\Omega_{k,\ell}$ are the complex coefficient and cos(AoD) of the $\ell$-th MPC of the channel vector for User $k$, respectively. We have $\sum \limits_{l=1}^{L_{k}}\mathbb{E}(|\lambda_{k,\ell}|^{2})\varpropto \frac{1}{d_{k}^{2}}$, where $d_{k}$ is the distance between the BS and User $k$.  $L_k$ is the total number of MPCs for User $k$, ${\bf{a}}(\cdot)$ is a steering vector function defined as
\begin{equation} \label{eq_steeringVCT}
\mathbf{a}(N,\Omega)=[e^{j\pi0\Omega},e^{j\pi1\Omega},e^{j\pi2\Omega},\cdot\cdot\cdot,e^{j\pi(N-1)\Omega}]^{\mathrm{T}},
\end{equation}
which depends on the array geometry. Let $\theta_{k,\ell}$ denote the real AoD of the $\ell$-th MPC for User $k$, then we have $\Omega_{k,\ell}=\cos(\theta_{k,\ell})$. Therefore, $\Omega_{k,\ell}$ is within the range $[-1, 1]$.

In general, the optimal decoding order of NOMA is the increasing order of the effective channel gains, i.e., $\left |\mathbf{h}_{k}^{\rm{H}}\mathbf{w} \right |^{2}$. However, we cannot determine the order of the effective channel gains before beamforming design. For simplicity, we utilize the increasing order of uses' channel gains as the decoding order. We will illustrate the rational of selecting the increasing-channel-gain decoding order in Section III-C, and verify that it can achieve near optimal performance by simulations. Without loss of generality, we assume  $\|\mathbf{h}_{1}\|\geq \|\mathbf{h}_{2}\|\geq \cdots \geq \|\mathbf{h}_{K}\|$. Therefore, User $k$ can decode $s_n~(k+1 \leq n \leq K)$ and then remove them from the received signal in a successive manner. The signals for User $m~(1\leq m \leq k-1)$ are treated as noise. Thus, the achievable rate of User $k$ is given by
\begin{equation}\label{eq_Rate}
R_{k}=\log_{2}(1+ \frac{\left |\mathbf{h}_{k}^{\rm{H}}\mathbf{w} \right |^{2}p_{k}}{\left |\mathbf{h}_{k}^{\rm{H}}\mathbf{w} \right |^{2}\sum \limits_{m=1}^{k-1}p_{m}+\sigma^{2}}).
\end{equation}

\subsection{Problem Formulation}
As aforementioned, both beamforming and power allocation have an important effect on the performance of the mmWave-NOMA system. To improve the overall data rate and guarantee user fairness, we formulate a problem to maximize the minimal achievable rate (the max-min fairness) among the $K$ users in this paper, where beamforming and power allocation are jointly optimized. The problem is formulated as
\begin{equation}\label{eq_problem}
\begin{aligned}
\mathop{\mathrm{Max}}\limits_{\{p_k\},\bf{w}}~ &\min\limits_{k}\{R_{k}\}\\
\mathrm{s.t.}~~~~ &C_1:~p_{k} \geq 0, ~~k=1,2,\cdots,K\\
&C_2:~\sum \limits_{k=1}^{K} p_{k} \leq P, \\
&C_3:~\|\bf{w}\|\leq 1,\\
&C_4:~|[{\mathbf{w}}]_i|=\frac{1}{\sqrt{N}} \text{ or } |[{\mathbf{w}}]_i| \leq \frac{2}{\sqrt{N}}, ~i=1,2,...,N
\end{aligned}
\end{equation}
where $R_k$ denotes the achievable rate of User $k$ as defined in \eqref{eq_Rate} and $\min\limits_{k}\{R_{k}\}$ is the minimal achievable rate among the $K$ served users. The constraint $C_1$ indicates that the power allocation to each user should be positive. $C_2$ is the transmission power constraint, where $P$ is the total transmission power. $C_3$ is the norm constraint on the AWV, and $C_4$ is the additional modulus constraint on the AWV for SPS or DPS implementation. The above problem is challenging, not only due to the non-convex formulation, but also due to that the variables to be optimized are entangled with each other. It is computationally prohibitive to directly search the optimal solution, because the dimension of the optimization variables is $N+K$, which is large in general. Next, we will propose a sub-optimal solution with promising performance but low computational complexity.

\section{Solution of the Problem}
As the modulus constraints for SPS and DPS implementations are different, we first solve the problem without considering the constraint $C_4$. As thus, Problem \eqref{eq_problem} is simplified as
\begin{equation}\label{eq_problem2}
\begin{aligned}
\mathop{\mathrm{Max}}\limits_{\{p_k\},\bf{w}}~ &\min\limits_{k}\{R_{k}\}\\
\mathrm{s.t.}~~~~ &C_1:~p_{k} \geq 0, ~~k=1,2,\cdots,K\\
&C_2:~\sum \limits_{k=1}^{K} p_{k} \leq P, \\
&C_3:~\|\bf{w}\|\leq 1.
\end{aligned}
\end{equation}
We will solve Problem \eqref{eq_problem2} first, and then particularly consider the modulus constraints in Section III-D.

Problem \eqref{eq_problem2} is still difficult due to the non-convex formulation, so we propose a sub-optimal solution with two stages. In the first stage, we obtain the closed-form optimal power allocation with an arbitrary fixed AWV. Then, in the second stage, we propose an appropriate beamforming algorithm utilizing the angle-domain spatial sparsity of the mmWave channel.
\subsection{Optimal Power Allocation with an Arbitrary Fixed AWV}
First, we introduce a variable to simplify Problem \eqref{eq_problem2}. Denote the minimal achievable rate among the $K$ users as $r$. Then Problem \eqref{eq_problem2} can be re-written as
\begin{equation}\label{eq_problem3}
\begin{aligned}
\mathop{\mathrm{Max}}\limits_{\{p_k\},\mathbf{w},r}~ &r\\
\mathrm{s.t.}~~~~ &C_0:~R_{k} \geq r, ~~k=1,2,\cdots,K\\
&C_1:~p_{k} \geq 0, ~~k=1,2,\cdots,K\\
&C_2:~\sum \limits_{k=1}^{K} p_{k} \leq P, \\
&C_3:~\|\bf{w}\|\leq 1,
\end{aligned}
\end{equation}
where the constraints $C_0~:~R_{k} \geq r,~(k=1,2,\cdots,K)$ are necessary and sufficient conditions of the fact that $r$ is the minimal achievable rate among the served users. On one hand, as $r$ is the minimal rate, the achievable rate of each user should be no less than $r$. On the other hand, there is at least one user, whose achievable rate $R_{k_{m}}$ is equal to $r$; otherwise we can always improve $r$ to minish the gap between $R_{k_{m}}$ and $r$.

We give the following Theorem to obtain the optimal solution of power allocation of Problem \eqref{eq_problem3} with an arbitrary fixed AWV.
\begin{theorem} Given an arbitrary fixed $\mathbf{w_{0}}$, the optimal power allocation of Problem \eqref{eq_problem3} is
\begin{equation}\label{power_criterion}
\left\{\begin{aligned}
&p_{1}=\eta\frac{\sigma^2}{\left |\mathbf{h}_{1}^{\rm{H}}\mathbf{w}_{0} \right |^{2}},\\
&p_{2}=\eta(p_{1}+\frac{\sigma^2}{\left |\mathbf{h}_{2}^{\rm{H}}\mathbf{w}_{0} \right |^{2}}),\\
&~~~~\vdots\\
&p_{K}=\eta(\sum \limits_{m=1}^{K-1} p_{m}+\frac{\sigma^2}{\left |\mathbf{h}_{K}^{\rm{H}}\mathbf{w}_{0} \right |^{2}}),
\end{aligned}\right.
\end{equation}
where $\eta=2^{r}-1$, and with the optimal power allocation, $R_{k}=r~(k=1,2,\cdots,K)$.
\end{theorem}

Before proving Theorem 1, we give Lemma 1 for the summation of the optimal power allocation in \eqref{power_criterion}, which is a function of $\eta$.

\begin{lemma} The summation of power allocation in \eqref{power_criterion} is
\begin{equation}\label{sum_power}
g(\eta)\triangleq \sum \limits_{k=1}^{K} p_{k}=\sum \limits_{k=1}^{K} \frac{\eta(1+\eta)^{K-k}\sigma^2}{\left |\mathbf{h}_{K}^{\rm{H}}\mathbf{w}_{0} \right |^{2}}.
\end{equation}
\end{lemma}
\begin{proof}
We prove Lemma 1 with mathematical induction.

When $m=1$, \eqref{sum_power} is easy to verify
\begin{equation}
p_{1}=\eta\frac{\sigma^2}{\left |\mathbf{h}_{1}^{\rm{H}}\mathbf{w}_{0} \right |^{2}}.
\end{equation}

When $m=n~(n\geq 1)$, assume that
\begin{equation}\label{k=n}
\sum \limits_{k=1}^{n} p_{k}=\sum \limits_{k=1}^{n} \frac{\eta(1+\eta)^{n-k}\sigma^2}{\left |\mathbf{h}_{k}^{\rm{H}}\mathbf{w}_{0} \right |^{2}}.
\end{equation}

When $m=n+1$, based on \eqref{k=n}, we have
\begin{equation}
\begin{aligned}
&\sum \limits_{k=1}^{n+1} p_{k}\\
=&\sum \limits_{k=1}^{n} p_{k}+\eta(\sum \limits_{k=1}^{n} p_{k}+\frac{\sigma^2}{\left |\mathbf{h}_{n}^{\rm{H}}\mathbf{w}_{0} \right |^{2}})\\
=&(1+\eta)\sum \limits_{k=1}^{n} p_{k}+\eta\frac{\sigma^2}{\left |\mathbf{h}_{n}^{\rm{H}}\mathbf{w}_{0} \right |^{2}}\\
=&(1+\eta)\sum \limits_{k=1}^{n} \frac{\eta(1+\eta)^{n-k}\sigma^2}{\left |\mathbf{h}_{k}^{\rm{H}}\mathbf{w}_{0} \right |^{2}}+\eta\frac{\sigma^2}{\left |\mathbf{h}_{n}^{\rm{H}}\mathbf{w}_{0} \right |^{2}}\\
=&\sum \limits_{k=1}^{n+1} \frac{\eta(1+\eta)^{n+1-k}\sigma^2}{\left |\mathbf{h}_{k}^{\rm{H}}\mathbf{w}_{0} \right |^{2}}.
\end{aligned}
\end{equation}

Finally, we can conclude that \eqref{sum_power} is true.
\end{proof}

Based on Lemma 1, the proof of Theorem 1 is presented in Appendix A. According to Theorem 1 and Lemma 1, Problem \eqref{eq_problem3} can be equivalently written as
\begin{equation}\label{beamforming}
\begin{aligned}
\mathop{\mathrm{Max}}\limits_{\mathbf{w},\eta}~~~ &\eta\\
\mathrm{s.t.}~~~~ &\sum \limits_{k=1}^{K} p_{k}=\sum \limits_{k=1}^{K} \frac{\eta(1+\eta)^{K-k}\sigma^2}{\left |\mathbf{h}_{k}^{\rm{H}}\mathbf{w} \right |^{2}}\leq P,\\
&\|\bf{w}\|\leq 1,
\end{aligned}
\end{equation}
where $\eta=2^r-1$.

Hereto, the first stage to solve Problem \eqref{eq_problem2} is finished, where the optimal power allocation is obtained, and thus the original problem with entangled power allocation and beamforming is reduced to a pure beamforming problem as shown in \eqref{beamforming}, which will be solved in the next subsection.

\subsection{Beamforming Design with Optimal Power Allocation}
The remaining task is to to solve Problem \eqref{beamforming} and obtain $\bf{w}$; then the closed-form expression of $\{p_{k}~(k=1,2,\cdots ,K)\}$ can be obtained by \eqref{power_criterion}. The main challenge is that the first constraint is non-convex, where $\bf{w}$ and $\eta$ are entangled. As the dimension of $\bf{w}$, i.e., $N$, is large in general, it is computationally prohibitive to directly search the optimal solution. However, the introduced variable $\eta$ is only 1-dimensional. We can search the maximal value of $\eta$ in the range of $[0,\Gamma]$ with the bisection method, where $\Gamma$ is the search upper bound. According to the definition of $\eta=2^r-1$, $\eta$ in fact represents the minimal signal to interference plus noise power ratio (SINR) among the $K$ users. If we allocate all the beam gain and power to the user with the best channel condition, i.e., User 1, then User 1 can achieve the highest SINR $\Gamma=(\sum \limits_{n=1}^{N}|[\mathbf{h}_{1}]_{n}|)^2P/(N\sigma^2)$. Thus, we select $\Gamma$ as the search upper bound. Given a fixed $\eta$, we judge whether an appropriate $\mathbf{w}$ can be found in the feasible region of Problem \eqref{beamforming}. Thus, we need to solve the following problem
\begin{equation}\label{beamforming2}
\begin{aligned}
\mathop{\mathrm{Min}}\limits_{\mathbf{w}}~~~~ &f(\mathbf{w})\triangleq \sum\limits_{k=1}^{K} \frac{\eta(1+\eta)^{K-k}\sigma^2}{\left |\mathbf{h}_{k}^{\rm{H}}\mathbf{w} \right |^{2}}\\
\mathrm{s.t.}~~~~ &\|\bf{w}\|\leq 1.
\end{aligned}
\end{equation}
Given $\eta$, if the minimal value of the objective function in Problem \eqref{beamforming2} is no larger than $P$, which means that a feasible solution can be found with the given $\eta$, we enlarge $\eta$ and solve Problem \eqref{beamforming2} again. If the minimal value of the objective function in Problem \eqref{beamforming2} is larger than $P$, i.e., a feasible solution cannot be found with the given $\eta$, we lessen $\eta$ and solve Problem \eqref{beamforming2} again. The stopping criterion of the bisection search is that $\eta$ meets an accuracy requirement.

To solve Problem \eqref{beamforming2}, some approximate manipulations are required to simplify the beamforming problem. Retrospecting the characteristic of the mmWave channel, the channel response vectors of different users are approximatively orthogonal due to the spatial sparsity in the angle domain, which is
\begin{equation}\label{orthogonal}
\frac{\mathbf{h}_{m}^{\rm{H}}}{\|\mathbf{h}_{m}^{\rm{H}}\|}\frac{\mathbf{h}_{n}}{\|\mathbf{h}_{n}\|}\approx
\left\{\begin{aligned}
&1,~\text{If}~m=n;\\
&0,~\text{If}~m\neq n.
\end{aligned}\right.
\end{equation}

With this approximation, $\{\frac{\mathbf{h}_{k}}{\|\mathbf{h}_{k}\|},~k=1,2,\cdots,K\}$ can be considered as an orthonormal basis of a subspace in $\mathbb{C}^{N}$. We say the subspace expanded by $\{\frac{\mathbf{h}_{k}}{\|\mathbf{h}_{k}\|},~k=1,2,\cdots,K\}$ is a \emph{channel space}. In Problem \eqref{beamforming2}, most beam gains are inclined to focus on the users' directions. Thus, the AWV should be located in the channel space, which can be written as
\begin{equation}\label{coordinates}
\mathbf{w}=\sum \limits_{k=1}^{K} \alpha_{k}\frac{\mathbf{h}_{k}}{\|\mathbf{h}_{k}\|},
\end{equation}
where $\{\alpha_{k},~k=1,2,\cdots,K\}$ are the coordinates of $\mathbf{w}$ in the channel space. Substituting \eqref{coordinates} into Problem \eqref{beamforming2}, we have
\begin{equation}\label{beamforming3}
\begin{aligned}
\mathop{\mathrm{Min}}\limits_{\{\alpha_{k}\}}~~~~ &\sum \limits_{k=1}^{K} \frac{\eta(1+\eta)^{K-k}\sigma^2}{\alpha_{k}^{2}\|\mathbf{h}_{k}\|^{2}}\\
\mathrm{s.t.}~~~~ &\sum \limits_{k=1}^{K} \alpha_{k}^{2}= 1.
\end{aligned}
\end{equation}

Note that the norm constraint for $\|\bf{w}\|\leq 1$ is replaced by $\|\bf{w}\|= 1$ here, because the norm of optimal $\bf{w}$ is surely 1. Assuming that $\bf{w}^{\star}$ is optimal and $\|\bf{w}^{\star}\|< 1$, we can always normalize the AWV to get a better solution of $\frac{\bf{w}^{\star}}{\|\bf{w}^{\star}\|}$.

To solve Problem \eqref{beamforming3}, we define the Lagrange function as
\begin{equation}\label{Lagrange}
L(\alpha,\lambda)=\sum \limits_{k=1}^{K} \frac{\eta(1+\eta)^{K-k}\sigma^2}{\alpha_{k}^{2}\|\mathbf{h}_{k}\|^{2}}+\lambda(\sum \limits_{k=1}^{K} \alpha_{k}^{2}-1).
\end{equation}

The Karush-Kuhn-Tucker (KKT) conditions can be obtained by the following equation \cite{boyd2004convex},
\begin{equation}\label{KKT}
\left\{\begin{aligned}
&\frac{\partial L}{\partial \alpha_{k}}=0, ~k=1,2,\cdots,K\\
&\frac{\partial L}{\partial \lambda}=0.
\end{aligned}\right.
\end{equation}

From the KKT conditions, we can obtain the solution of Problem \eqref{beamforming3}, which is given by
\begin{equation}
\begin{aligned}
&\frac{\partial L}{\partial \alpha_{k}}=0\\
\Rightarrow &\frac{-2\eta(1+\eta)^{K-k}\sigma^2}{\alpha_{k}^{3}\|\mathbf{h}_{k}\|^{2}}+2\lambda\alpha_{k}=0\\
\Rightarrow &\alpha_{k}=\sqrt[4]{\frac{\eta(1+\eta)^{K-k}\sigma^2}{\lambda\|\mathbf{h}_{k}\|^{2}}}\\
\Rightarrow &\alpha_{k} \propto \sqrt[4]{\frac{\eta(1+\eta)^{K-k}}{\|\mathbf{h}_{k}\|^{2}}}.
\end{aligned}
\end{equation}

Thus, the designed AWV in Problem \eqref{beamforming2} is given by
\begin{equation}\label{BF_vector}
\left\{\begin{aligned}
&\mathbf{\bar{w}}=\sum \limits_{k=1}^{K} \sqrt[4]{\frac{\eta(1+\eta)^{K-k}}{\|\mathbf{h}_{k}\|^{2}}}\frac{\mathbf{h}_{k}}{\|\mathbf{h}_{k}\|},\\
&\mathbf{w}=\frac{\mathbf{\bar{w}}}{\|\mathbf{\bar{w}}\|}.
\end{aligned}\right.
\end{equation}

In summary, we give Algorithm 1 to solve Problem \eqref{beamforming}.
\begin{algorithm}[h]
\caption{AWV design}
\label{alg1}
\begin{algorithmic}[1]
\REQUIRE ~\\
Channel response vectors: $\mathbf{h}_{k}, ~k=1,2,\cdots,K$;\\
Total transmission power: $P$;\\
Noise power: $\sigma^2$;\\
The search accuracy $\epsilon$.\\
\ENSURE ~\\
$\eta$ and $\mathbf{w}$.\\
\STATE $\eta_{\min}=0,~\eta_{\max}=\Gamma$.
\WHILE {$\eta_{\max}-\eta_{\min}>\epsilon$}
\STATE $\eta=(\eta_{\max}+\eta_{\min})/2$;
\STATE Calculate $\mathbf{w}$ according to \eqref{BF_vector} and the objective function in Problem \eqref{beamforming2}: $f(\mathbf{w})$.
\IF{$f(\mathbf{w})>P$}
\STATE $\eta_{\max}=\eta$.
\ELSE
\STATE $\eta_{\min}=\eta$.
\ENDIF
\ENDWHILE
\RETURN $\eta$ and $\mathbf{w}$.
\end{algorithmic}
\end{algorithm}

Hereto, we have solved Problem \eqref{eq_problem2} and obtain the solution $\{p_{k}^{\star}, \mathbf{w}\}$, where the AWV is obtained in Algorithm 1 and the power allocation is given in \eqref{power_criterion}. The AWV is approximately optimal while the power allocation is optimal for the designed AWV. A leftover problem is to verify the rational of the decoding order. We will consider this problem next.

\subsection{Decoding order}
When formulating Problem \eqref{eq_problem}, we assumed that the decoding order of signals is the increasing order of the channel gains. Next, we will verify that the order of the effective channel gains after beamforming design is the same with the channel-gain order. The effective channel gain for User $k$ is
\begin{equation}\label{beam_gain}
\begin{aligned}
&|\mathbf{h}_{k}^{\rm{H}}\mathbf{w}|^{2} \propto |\mathbf{h}_{k}^{\rm{H}}\mathbf{\bar{w}}|^{2}\\
=&\Bigg{|}\sum \limits_{m=1}^{K} \sqrt[4]{\frac{\eta(1+\eta)^{K-m}}{\|\mathbf{h}_{m}\|^{2}}}\frac{\mathbf{h}_{k}^{\rm{H}}\mathbf{h}_{m}}{\|\mathbf{h}_{m}\|}\Bigg{|}^{2}\\
\substack{{(a)}\\=}&\Bigg{|} \sqrt[4]{\frac{\eta(1+\eta)^{K-k}}{\|\mathbf{h}_{k}\|^{2}}}\frac{\mathbf{h}_{k}^{\rm{H}}\mathbf{h}_{k}}{\|\mathbf{h}_{k}\|}\Bigg{|}^{2}\\
=&\sqrt{\eta(1+\eta)^{K-k}}\|\mathbf{h}_{k}\|,
\end{aligned}
\end{equation}
where $(a)$ is according to the orthogonal assumption of the channel response vectors. As $\eta=2^{r}-1>0$, $\sqrt{\eta(1+\eta)^{K-k}}$ is decreasing for $k$. We have assumed that the order of the users' channel gains is $\|\mathbf{h}_{1}\|\geq \|\mathbf{h}_{2}\|\geq \cdots \geq \|\mathbf{h}_{K}\|$. Thus, under the orthogonal assumption of the channel response vectors, the order of users' effective channel gains is
\begin{equation}\label{beam_gain_order}
|\mathbf{h}_{1}^{\rm{H}}\mathbf{w}|^{2}\geq |\mathbf{h}_{2}^{\rm{H}}\mathbf{w}|^{2}\geq \cdots |\mathbf{h}_{K}^{\rm{H}}\mathbf{w}|^{2}.
\end{equation}

As shown in \eqref{beam_gain_order}, the order of the effective channel gains is the same with that of channel gains. However, this property may not hold if we utilize other decoding orders, which indicates that the increasing-channel-gain decoding order is more reasonable. In the simulations, we will compare the performance of different decoding orders and find that the performance of increasing-channel-gain decoding order is very close to the performance of the optimal decoding order.

\subsection{Consideration of Modulus Constraints}
When solving Problem \eqref{eq_problem2}, the additional modulus constraints on the AWV were not considered. Next, we will consider the modulus constraints and solve the original problem, i.e., Problem \eqref{eq_problem}. As we have shown in the system model, the modulus constraints on the elements of the AWV are \eqref{CM} and \eqref{DPS} for SPS and DPS implementations, respectively. Some additional normalized operations on the designed AWV are required to satisfy the constraints. For the SPS implementation, the constant modulus normalization is given by
\begin{equation}\label{CM_normalization}
[\mathbf{w}_{S}]_{i}=\frac{[\mathbf{w}]_{i}}{\sqrt{N}\big{|}[\mathbf{w}]_{i}\big{|}}, ~i=1,2,\cdots,N.
\end{equation}
where $\mathbf{w}_{S}$ denotes the AWV for SPS implementation. For the DPS implementation, the modulus normalization is given by
\begin{equation}\label{DPS_normalization}
[\mathbf{w}_{D}]_{i}=
\left\{\begin{aligned}
&[\mathbf{w}]_{i}, ~\text{If}~\big{|}[\mathbf{w}]_{i}\big{|}\leq \frac{2}{\sqrt{N}};\\
&\frac{2}{\sqrt{N}}, ~\text{If}~\big{|}[\mathbf{w}]_{i}\big{|}>\frac{2}{\sqrt{N}}.
\end{aligned}\right.
\end{equation}
where $\mathbf{w}_{D}$ denotes the AWV for DPS implementation. Each element of $\mathbf{w}_{D}$ is the sum weight of the corresponding antenna branch, and it needs to be decomposed into two components, which can be expressed as
\begin{equation}
[\mathbf{w}_{D}]_{i}\triangleq a_{i}e^{j\theta_{i}}=\frac{1}{\sqrt{N}}e^{j(\theta_{i}+\varphi_{i})}+\frac{1}{\sqrt{N}}e^{j(\theta_{i}-\varphi_{i})},
\end{equation}
where $a_{i}\in [0,\frac{2}{\sqrt{N}}]$ and $\theta_{i}\in [0,2\pi)$ are the modulus and the phase of $[{\mathbf{w}}_D]_{i}$, respectively, and $\varphi_i=\arccos(\frac{\sqrt{N}a_i}{2})$. Thus, the weights of the two PSs corresponding to $[\mathbf{w}_{D}]_{i}$ are
\begin{equation}
\left\{\begin{aligned}
&[\mathbf{\tilde{w}}_{D}]_{2i-1}=\frac{1}{\sqrt{N}}e^{j(\theta_{i}+\varphi_{i})},\\
&[\mathbf{\tilde{w}}_{D}]_{2i}=\frac{1}{\sqrt{N}}e^{j(\theta_{i}-\varphi_{i})}.
\end{aligned}\right.
\end{equation}

\subsection{Computational Complexity}
As we obtained the closed-form optimal power allocation with an arbitrary fixed AWV, the computational complexity is mainly caused by the beamforming algorithm in the second stage. In Algorithm 1, the total search time for $\eta$ is $T=\log_{2}(\frac{\Gamma}{\epsilon})$, where $\Gamma$ is the search upper bound and  $\epsilon$ is the search accuracy. Thus, the computational complexity of the proposed method is $\mathcal{O}(T)$, which does not increase as $N$ and $K$. However, if we directly search the solution of Problem \eqref{eq_problem} and obtain the globally optimal solution, the total complexity is $\mathcal{O}((\frac{1}{\epsilon})^{N+K})$, which exponentially increases as $N$ and $K$.

\section{Performance Simulations}
In this section, we provide simulation results to verify the performance of the proposed joint beamforming and power allocation method in the mmWave-NOMA system. We adopt the channel model in \eqref{eq_oriChannel} in the simulations, where the users are uniformly distributed from 10m to 500m away from the BS, and the channel gain of the user 100m away from the BS has an average power of 0dB. The number of MPCs for all the users are $L=4$. Both LOS and NLOS channel models are considered. For the LOS channel, the average power of the NLOS paths is 15 dB weaker than that of the LOS path. For the NLOS channel, the coefficient of each path has an average power of $1/\sqrt{L}$. The search accuracy in Algorithm 1 is $\epsilon=10^{-6}$.

\begin{figure}[t]
\begin{center}
  \includegraphics[width=\figwidth cm]{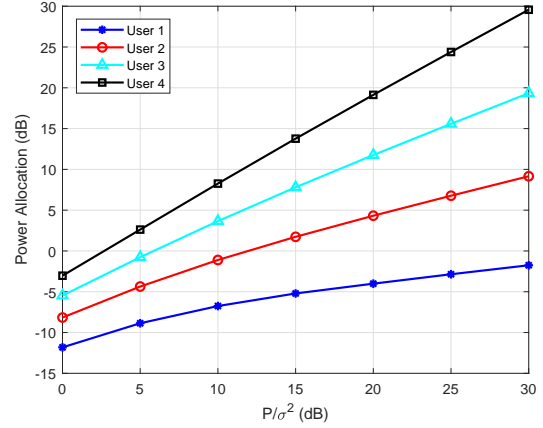}
  \caption{Power allocation with varying total power to noise ratio, where $N=32$ and $K=4$.}
  \label{fig:Beam_gain_P}
\end{center}
\end{figure}
\begin{figure}[t]
\begin{center}
  \includegraphics[width=\figwidth cm]{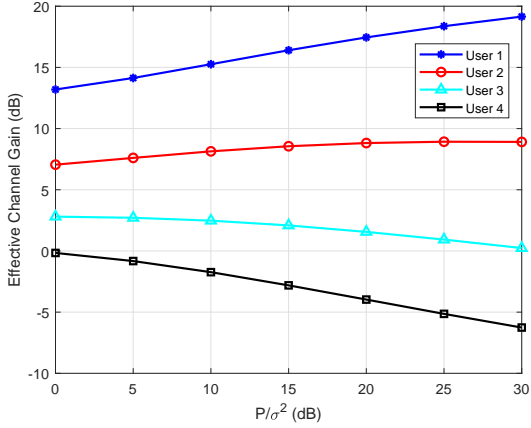}
  \caption{Effective channel gains with varying total power to noise ratio, where $N=32$ and $K=4$.}
  \label{fig:Power_allocation_P}
\end{center}
\end{figure}
We first show the power allocation and the effective channel gains in Figs. \ref{fig:Beam_gain_P} and \ref{fig:Power_allocation_P}, respectively, where the LOS channel model is adopted \footnote{Similar results can be observed when the NLOS channel model is adopted; thus the results are not presented here for conciseness.}. Each point is an average result from $10^4$ channel realizations. From Fig. \ref{fig:Beam_gain_P} we can find that most power is allocated to User 4, the user with the lowest channel gain. Less power is allocated to the users with higher channel gains, so as to reduce interference. Despite all this, it can be observed from Fig. \ref{fig:Power_allocation_P} that the effective channel gain of User 4 is still the lowest. The user with a better channel gain have a higher effective channel gain with the proposed solution, which verifies the conclusion in Section III-C about the decoding order. It is noteworthy that the effective channel gains of User 1 and User 4 go increasing and decreasing, respectively, when $P/\sigma^2$ becomes higher, which is the result of joint power allocation and beamforming. It indicates that when the total power is high, power and beam gain should be jointly allocated to enlarge the difference of the effective channel gains to achieve a larger minimal user rate.

Next, we compare the performance between the considered mmWave-NOMA system and a mmWave-OMA system. We give the following method to calculate the minimal achievable rates in a $K$-user mmWave-OMA system, where time division multiple access (TDMA) is used without of generality.

If all the time slots are allocated to User $k$, the achievable rate for User $k$ is
\begin{equation}
\bar{R}_{k}=\log_{2}(1+ \frac{\left |\mathbf{h}_{k}^{\rm{H}}\mathbf{w} \right |^{2}P}{\sigma^{2}}).
\end{equation}

Assume that the time division is ideal, which means that the time slot can be allocated to the users with any proportion. To maximize the minimal achievable rate of the $K$ users, more time should be allocated to the users with lower channel gains, such that the achievable rates of the $K$ users are equal. Thus, the time allocation for User $k$ is
\begin{equation}\label{OMA_allo}
\beta_{k}=\frac{1/\bar{R}_{k}}{\sum \limits_{m=1}^{K} 1/\bar{R}_{m}}.
\end{equation}

Then the achievable rate of User $k$ in the mmWave-OMA system is
\begin{equation}\label{OMA_Rate}
R_{k}^{\text{OMA}}=\beta_{k}\bar{R}_{k}=\frac{1}{\sum \limits_{m=1}^{K} 1/\bar{R}_{m}},
\end{equation}
where all the users have the same achievable rate.

\begin{figure}[t]
\begin{center}
  \includegraphics[width=\figwidth cm]{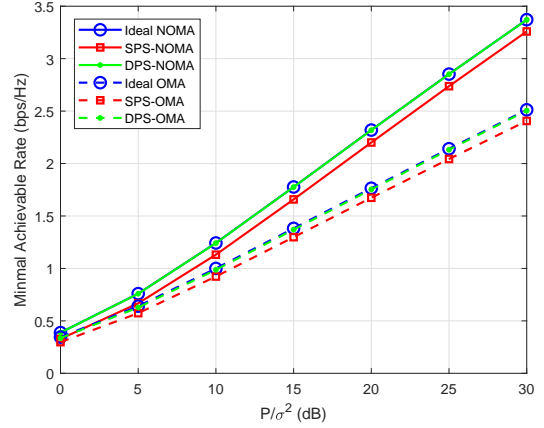}
  \caption{Comparison of the minimal achievable rates between NOMA and OMA system with varying total power to noise ratio, where $N=32$ and $K=4$.}
  \label{fig:Rate_P}
\end{center}
\end{figure}

Fig. \ref{fig:Rate_P} shows the comparison result of the minimal achievable rates between the mmWave-NOMA and mmWave-OMA systems with varying total power to noise ratio. The minimal achievable rates of Ideal NOMA/OMA, SPS-NOMA/SPS-OMA and DPS-NOMA/DPS-OMA are based on the beamforming given in \eqref{BF_vector}, \eqref{CM_normalization} and \eqref{DPS_normalization}, which are corresponding to the beamforming without CM constraint, with SPS implementation and with DPS implementation, respectively. Each point in the figure is the average performance of $10^4$ LOS channel realizations. We can find that the minimal achievable rates of SPS-NOMA are lower than that of DPS-NOMA, which is very close to the minimal achievable rates of Ideal NOMA, this is because the strict modulus normalization on the AWV for SPS results in significant performance loss, while the modulus normalization on the AWV for DPS is more relaxed and has little impact on the rate performance. In addition, the minimal achievable rates of the mmWave-NOMA system is distinctly better than those of the  mmWave-OMA system for all the cases, and superiority is more significant when the total power to noise ratio is higher.

\begin{figure}[t]
\begin{center}
  \includegraphics[width=\figwidth cm]{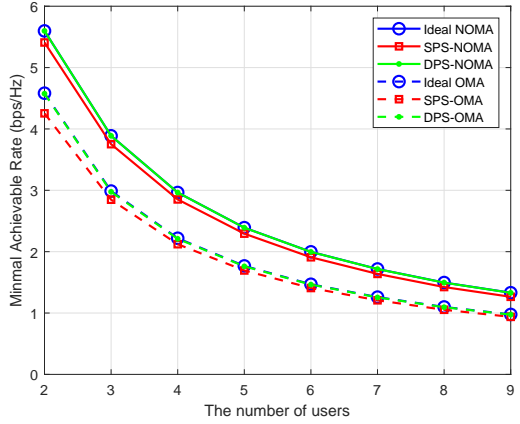}
  \caption{Comparison of the minimal achievable rates between the NOMA and OMA systems with varying number of users, where $N=32$ and the average transmission power to noise for each user is 20 dB.}
  \label{fig:Rate_K}
\end{center}
\end{figure}

Fig. \ref{fig:Rate_K} compares the minimal achievable rates between mmWave-NOMA and mmWave-OMA systems with varying number of users. For fairness, the total transmission power is proportional to the number of users, and the average transmission power to noise for each user is 20 dB. Each point in Fig. \ref{fig:Rate_K} is the average performance of $10^4$ LOS channel realizations. It can be observed again that the minimal achievable rate of mmWave-NOMA is better than that of mmWave-OMA for both SPS and DPS implementations, and the minimal achievable rates of DPS-NOMA is very close to that of Ideal NOMA. On the other hand, the minimal achievable rates of both mmWave-NOMA and mmWave-OMA decreases as the number of users increases. This is mainly due to that the orthogonality of the channel vectors of the users become weakened, which deteriorates the beamforming performance and in turn the minimal achievable rate performance.

\begin{figure}[t]
\begin{center}
  \includegraphics[width=\figwidth cm]{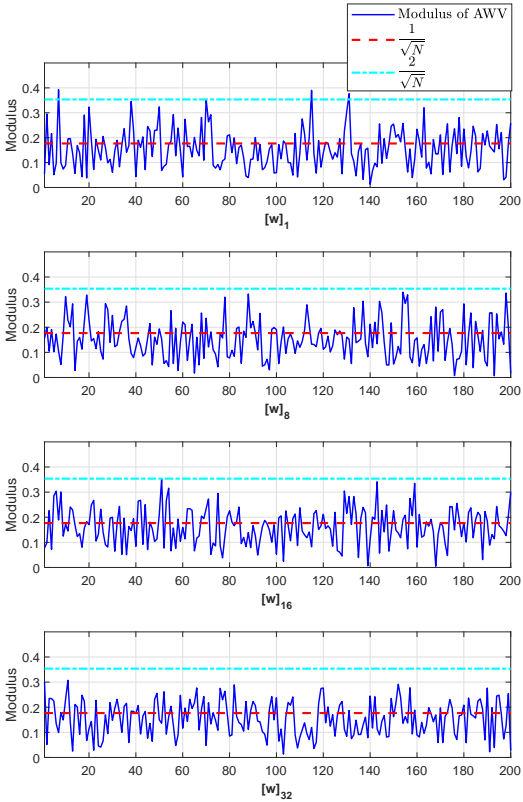}
  \caption{Moduli of the elements of the AWVs, where $N=32$, $K=4$ and $P/\sigma^2=25$ dB.}
  \label{fig:AWV_Modulus}
\end{center}
\end{figure}
Fig. \ref{fig:AWV_Modulus} shows the modulus of the elements of AWVs , where $N=32$, $K=4$ and $P/\sigma^2=25$ dB. We show the 1st, 8th, 16th and 32th element of 200 AWVs with different channel realizations. It can be seen that the moduli of the AWV's elements are mainly distributed around $1/\sqrt{N}$, and almost all of them have a modulus less than $2/\sqrt{N}$. The results in Fig. \ref{fig:AWV_Modulus} demonstrate that the modulus normalization for the DPS implementation has a limited impact on the performance.

In the second stage of the proposed solution, we have assumed that the channel response vectors are orthogonal and then found an appropriate AWV in \eqref{beamforming2}. To evaluate the impact of this approximation, we compare the performance of the proposed solution with the upper-bound performance. We solve Problem \eqref{beamforming2} using particle swarm optimization, where the density of particles is sufficiently high, and thus the obtained minimal achievable rate can be treated as the upper bound. Limited by the computational complexity, we provide the simulation results with a relatively small-scale antenna array, i.e., $N=8,16$.  The comparison result is shown in Fig. \ref{fig:RateBound_P}, where each point is averaged from $10^3$ LOS channel realizations. The minimal achievable rate of Ideal NOMA is based on the beamforming given in \eqref{BF_vector}, which is corresponding to the beamforming without the CM constraint and the orthogonality assumption of the channel vectors between the NOMA users. As we can see, when $N=8$, the performance gap between the proposed solution and the upper bound is no more than 0.25 bps/Hz. When $N=16$, the performance gap is even smaller, i.e., no more than 0.2 bps/Hz. The reason is that the orthogonality of the channel vectors becomes stronger when $N$ is larger. Thus, the approximation of the beamforming design in Problem \eqref{beamforming2} has limited impact on the system performance, and the proposed sub-optimal solution can achieve an near-upper-bound performance, especially when $N$ is large.

\begin{figure}[t]
\begin{center}
  \includegraphics[width=\figwidth cm]{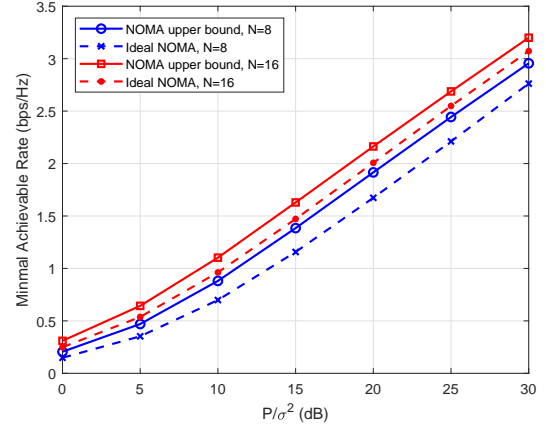}
  \caption{Comparison of the minimal achievable rates between the proposed solution and the upper bound with varying total power to noise ratio, where $K=4$.}
  \label{fig:RateBound_P}
\end{center}
\end{figure}

\begin{figure}[t]
\begin{center}
  \includegraphics[width=\figwidth cm]{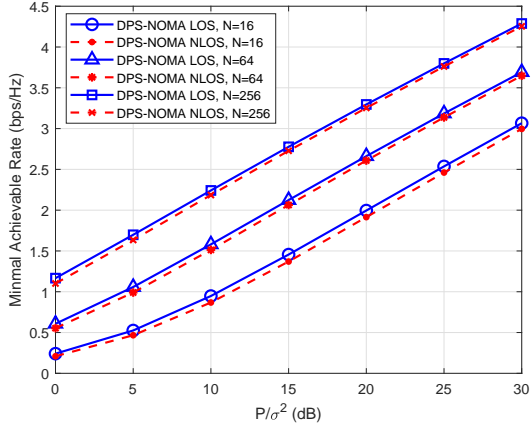}
  \caption{Performance comparison between LOS and NLOS channel models with varying total power to noise ratio, where $K=4$.}
  \label{fig:Rate_LOS_NOLS}
\end{center}
\end{figure}
Fig. \ref{fig:Rate_LOS_NOLS} compares the minimal achievable rates of mmWave-NOMA under the LOS and NLOS channel models with varying total power to noise ratio. The number of antennas is $N=16, 64, 256$, respectively. The number of users is $K=4$. Each point in Fig. \ref{fig:Rate_LOS_NOLS} is the average performance of $10^4$ channel realizations. It can be seen that the performance of DPS-NOMA with the LOS channel model is slightly better than that with the NLOS channel model, because the channel power is more centralized for the LOS channel. However, the performance gap between them is quite small, especially when $N$ is large. The reason is that according to \eqref{beam_gain}, the effective channel gain is linear to $\|\mathbf{h}_{k}\|$, the norm of the channel vector, rather than that of the power of the strongest path. Thus, the performance gap of DPS-NOMA with the LOS and NLOS channel models is small.

\begin{figure}[t]
\begin{center}
  \includegraphics[width=\figwidth cm]{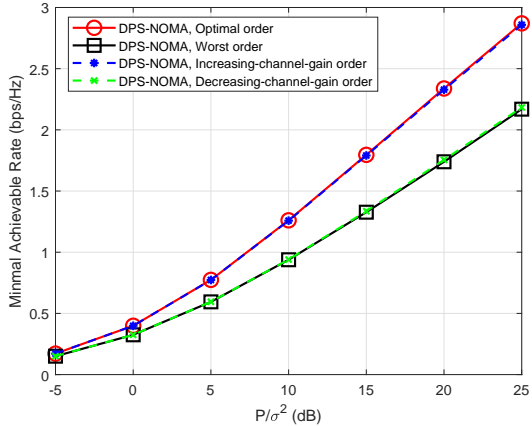}
  \caption{Comparison of the minimal achievable rates under different decoding orders with varying total power to noise ratio, where $N=32$ and $K=4$.}
  \label{fig:Rate_decoding_order}
\end{center}
\end{figure}

The simulations above are all based on the increasing-channel-gain decoding order. Next, we will show the impact of the decoding order on the mmWave-NOMA system. Fig. \ref{fig:Rate_decoding_order} shows the performance comparison between different decoding orders with varying total power to noise ratio, where $N=32$ and $K=4$. There are 24 decoding orders in total for the 4 users. Each point in Fig. \ref{fig:Rate_decoding_order} is the average performance of $10^4$ LOS channel realizations. The minimal achievable rates of the 24 decoding orders are all calculated. The order with the highest minimal achievable rate is chosen as the optimal order and the order with the lowest minimal achievable rate is chosen as the worst order. The increasing-channel-gain order is the one adopted in our solution, while the decreasing-channel-gain order is one for comparison. From the figure we can find that there is a significant performance gap between the optimal order and the worst order, which means that the decoding order has an important impact on the performance of mmWave-NOMA. Moreover, the performance with the increasing-channel-gain order is almost the same as the optimal one, while the performance with the decreasing-channel-gain order is almost the same as the worst one. This result shows the rational of adopting the increasing-channel-gain order in our solution.

\section{Conclusion}
In this paper, we have investigated downlink max-min fairness mmWave-NOMA with analog beamforming. A joint beamforming and power allocation problem was formulated and solved in two stages. In the first stage, the closed-form optimal power allocation was obtained with an arbitrary fixed AWV, reducing the joint beamforming and power allocation problem into an equivalent beamforming problem. Then, an appropriate beamforming vector was obtained by utilizing the spatial sparsity in the angle domain of the mmWave channel. Both implementations of SPS and DPS were considered with different modulus normalizations. The simulation results demonstrate that the modulus normalization has limited impact on the achievable rate performance, especially for the DPS implementation. Moreover, by using the proposed solution, the considered mmWave-NOMA system can achieve a near-upper-bound performance of the minimal achievable rate, which is significantly better than that of the conventional mmWave-OMA system.

\appendices
\section{Proof of Theorem 1}

Without loss of generality, we denote $\{p_{k}^{\star},r^{\star}\}$ one optimal solution of Problem \eqref{eq_problem3} with fixed $\mathbf{w_{0}}$, where the achievable rate of User $k$ is $R_{k}^{\star}$, and let $\eta^{\star}=2^{r^\star}-1$.

With $\eta^{\star}$ we can obtain another solution $\{p_{k}^{\circ},r^{\star}\}$ , where
\begin{equation}\label{power_allo}
\left\{\begin{aligned}
&p_{1}^{\circ}=\eta^{\star}\frac{\sigma^2}{\left |\mathbf{h}_{1}^{\rm{H}}\mathbf{w}_{0} \right |^{2}},\\
&p_{2}^{\circ}=\eta^{\star}(p_{1}^{\circ}+\frac{\sigma^2}{\left |\mathbf{h}_{2}^{\rm{H}}\mathbf{w}_{0} \right |^{2}}),\\
&~~~~\vdots\\
&p_{K}^{\circ}=\eta^{\star}(\sum \limits_{m=1}^{K-1} p_{m}^{\circ}+\frac{\sigma^2}{\left |\mathbf{h}_{K}^{\rm{H}}\mathbf{w}_{0} \right |^{2}}).
\end{aligned}\right.
\end{equation}
The following lemma shows that this solution is also an optimal one.

\begin{lemma}
The solution $\{p_{k}^{\circ},r^{\star}\}$ is also an optimal solution of Problem \eqref{eq_problem3}, and the achievable rates under this parameter setting always satisfy $R_{k}^{\circ}=r^{\star}~(1\leq k\leq K)$.
\end{lemma}
\begin{proof}
First, we need to verify that the constraints $C_{0}$, $C_{1}$ and $C_{2}$ are all satisfied.

According to the expression of \eqref{power_allo}, it is obvious that $\{p_{k}^{\circ}\geq 0\}$, which means that the constraint $C_{1}$ is satisfied.

In addition, according to the assumption that $\{p_{k}^{\star},r^{\star}\}$ is an optimal solution, we have
\begin{equation}\label{C0}
\begin{aligned}
&r^{\star} \leq R_{k}^{\star} \\
\Rightarrow &\eta^{\star} \leq \frac{\left |\mathbf{h}_{k}^{\rm{H}}\mathbf{w}_{0} \right |^{2}p_{k}^{\star}}{\left |\mathbf{h}_{k}^{\rm{H}}\mathbf{w}_{0} \right |^{2}\sum \limits_{m=1}^{k-1}p_{m}^{\star}+\sigma^{2}} \\
\Rightarrow &\eta^{\star}(\sum \limits_{m=1}^{k-1} p_{m}^{\star}+\frac{\sigma^2}{\left |\mathbf{h}_{k}^{\rm{H}}\mathbf{w}_{0} \right |^{2}}) \leq p_{k}^{\star}.
\end{aligned}
\end{equation}
Next, we use mathematical induction to prove that $p_{k}^{\circ}  \leq p_{k}^{\star} ~(k=1,2,\cdots,K)$.

When $k=1$, according to \eqref{C0} we have
\begin{equation}\label{eq_pk1}
p_{1}^{\circ} \leq p_{1}^{\star}.
\end{equation}

When $k=n~(n\geq1)$, assume $\{p_{1}^{\circ} \leq p_{1}^{\star},\cdots, p_{n}^{\circ} \leq p_{n}^{\star}\}$. According to \eqref{C0} we have
\begin{equation}\label{eq_pk2}
\begin{aligned}
&p_{n+1}^{\circ}=\eta^{\star}(\sum \limits_{m=1}^{n} p_{m}^{\circ}+\frac{\sigma^2}{\left |\mathbf{h}_{n+1}^{\rm{H}}\mathbf{w}_{0} \right |^{2}})\\
&\leq \eta^{\star}(\sum \limits_{m=1}^{n} p_{m}^{\star}+\frac{\sigma^2}{\left |\mathbf{h}_{n+1}^{\rm{H}}\mathbf{w}_{0} \right |^{2}})\leq p_{n+1}^{\star}.
\end{aligned}
\end{equation}

Thus, we can conclude that $p_{k}^{\circ}  \leq p_{k}^{\star} ~(k=1,2,\cdots,K)$ and we have
\begin{equation}
\sum \limits_{k=1}^{K} p_{k}^{\circ} \leq \sum \limits_{k=1}^{K} p_{k}^{\star} \leq P,
\end{equation}
which means that the constraint $C_2$ is satisfied.

With the considered solution $(p_{k}^{\circ},r^\star)$, we have
\begin{equation}
\begin{aligned}
R_{k}^{\circ}&=\log_{2}(1+ \frac{\left |\mathbf{h}_{k}^{\rm{H}}\mathbf{w}_{0} \right |^{2}p_{k}^{\circ}}{\left |\mathbf{h}_{k}^{\rm{H}}\mathbf{w}_{0} \right |^{2}\sum \limits_{m=1}^{k-1}p_{m}^{\circ}+\sigma^{2}})\\
&=\log_{2}(1+ \frac{p_{k}^{\circ}}{\sum \limits_{m=1}^{k-1}p_{m}^{\circ}+\frac{\sigma^{2}}{\left |\mathbf{h}_{k}^{\rm{H}}\mathbf{w}_{0} \right |^{2}}})\\
&\substack{{(a)}\\=}\log_{2}(1+\eta^{\star})\\
&=r^{\star},
\end{aligned}
\end{equation}
where $(a)$ is based on \eqref{power_allo}. The above equation means that the constraint $C_0$ is satisfied.

Since $\{p_{k}^{\circ},r^{\star}\}$ can satisfy all the constraints, and $R_{k}^{\circ}=r^{\star}~(1\leq k\leq K)$, it is also an optimal solution of Problem \eqref{eq_problem3}.
\end{proof}

As both $\{p_{k}^{\circ},r^{\star}\}$ and $\{p_{k}^{\star},r^{\star}\}$ are optimal solutions of Problem \eqref{eq_problem3}, we will prove that they are in fact the same as each other. For this sake, we need to prove that $R_{k}^{\star}=r^{\star}~(1\leq k\leq K)$. We assume that there exists one user whose achievable is strictly larger than $r^{\star}$, i.e., $R_{k_{0}}^{\star}>r^{\star}$, and we will prove that this assumption does not hold as follows.

As we have assumed that $R_{k_{0}}^{\star}>r^{\star}$, we have $R_{k_{0}}^{\star}>R_{k_{0}}^{\circ}=r^{\star}$. In addition, we have proven that $p_{k}^{\circ} \leq p_{k}^{\star}$ (see the proof in \eqref{eq_pk1} and \eqref{eq_pk2}). According to the expression of $R_{k}$ in \eqref{eq_Rate}, it is straightforward to derive $p_{k_{0}}^{\star}>p_{k_{0}}^{\circ}$.

We define another solution $\{p_{k}^{\vartriangle},r^{\vartriangle}\}$, where $r^{\vartriangle}=r^{\star}+\delta$, and
\begin{equation}\label{power_allo2}
\left\{\begin{aligned}
&p_{1}^{\vartriangle}=\eta^{\vartriangle}\frac{\sigma^2}{\left |\mathbf{h}_{1}^{\rm{H}}\mathbf{w}_{0} \right |^{2}},\\
&p_{2}^{\vartriangle}=\eta^{\vartriangle}(p_{1}^{\vartriangle}+\frac{\sigma^2}{\left |\mathbf{h}_{2}^{\rm{H}}\mathbf{w}_{0} \right |^{2}}),\\
&~~~~\vdots\\
&p_{K}^{\vartriangle}=\eta^{\vartriangle}(\sum \limits_{m=1}^{K-1} p_{m}^{\vartriangle}+\frac{\sigma^2}{\left |\mathbf{h}_{K}^{\rm{H}}\mathbf{w}_{0} \right |^{2}}),
\end{aligned}\right.
\end{equation}
where $\eta^{\vartriangle}=2^{r^{\vartriangle}}-1$ and $\delta>0$. Thus, we have $\eta^{\vartriangle}>\eta^{\star}$.

Next, we prove that $\{p_{k}^{\vartriangle},r^{\vartriangle}\}$ is within the feasible region of Problem \eqref{eq_problem3}. Similar to the proof in Lemma 2, we can prove that $\{p_{k}^{\vartriangle}\geq 0\}$ and $R_{k}^{\vartriangle}=r^{\vartriangle}> r^{\star}~(1\leq k\leq K)$, which means that the constraints $C_{0}$ and $C_{1}$ are satisfied. According to Lemma 1, the summation of power allocation in \eqref{power_allo} and \eqref{power_allo2} are $g(\eta^{\star})$ and $g(\eta^{\vartriangle})$, respectively. As we have proven that $p_{k_{0}}^{\star}>p_{k_{0}}^{\circ}$, we have $g(\eta^{\star})<P$. Otherwise, if $g(\eta^{\star})=P$, $\sum \limits_{k=1}^{K} p_{k}^{\star}>\sum \limits_{k=1}^{K} p_{k}^{\circ}=g(\eta^{\star}) =P$, which is contradictory to Constraint $C_{2}$ in Problem \eqref{eq_problem3}. As $g(\eta)$ is an increasing function for $\eta$, we can always find a small positive $\delta$, which satisfies $g(\eta^{\star}+\delta)<P$, i.e., $g(\eta^{\vartriangle})<P$. Thus, the constraint $C_{2}$ is satisfied with sufficiently small $\delta$.

In brief, $\{p_{k}^{\vartriangle},r^{\vartriangle}\}$ is within the feasible region of Problem \eqref{eq_problem3} provided that $\delta$ is small enough. However, we have $R_{k}^{\vartriangle}=r^{\vartriangle}> r^{\star}~(1\leq k\leq K)$, which means that the solution $\{p_{k}^{\vartriangle},r^{\vartriangle}\}$ is better than $\{p_{k}^{\star},r^{\star}\}$, which is contradictory to the fact that $\{p_{k}^{\star},r^{\star}\}$ is an optimal solution. Thus, the assumption that there exists one user whose achievable is strictly larger than $r^{\star}$ does not hold. Equivalently, the achievable rates of users under the optimal power allocation satisfy $R_{k}^\star=r^\star=R_{k}^{\circ}~(1\leq k \leq K)$. Solve the equations set above and we can obtain that $\{p_{k}^{\star},r^{\star}\}$ is the same as $\{p_{k}^{\circ},r^{\star}\}$, and the optimal power allocation of Problem \eqref{eq_problem3} is given by \eqref{power_criterion}.


\end{document}